\documentclass[12pt]{article}
\usepackage{amsmath}
\usepackage{graphicx,psfrag,epsf}
\usepackage{enumerate}
\usepackage{psfrag,epsf}
\usepackage{url} 
\usepackage{amsmath,amssymb,amsfonts,amsthm,mathtools,mathrsfs}

\usepackage{booktabs}     
\usepackage[table]{xcolor} 
\usepackage{caption}      
\usepackage{array}        

\usepackage{booktabs}
\usepackage{bm,bbm}
\usepackage{color}
\usepackage{subfigure}
\usepackage{algorithm,algpseudocode}
\usepackage[symbol]{footmisc}
\usepackage{scalefnt}
\usepackage{authblk} 
\usepackage{multirow,centernot}
\usepackage[colorlinks=true, citecolor=blue, urlcolor=blue]{hyperref}
\usepackage{natbib}
\usepackage{dsfont}
\usepackage[title]{appendix}
\input cyracc.def

\usepackage[left=1in,top=1.1in,right=0.5in,bottom=1in]{geometry}

\theoremstyle{definition}

\newtheorem{theorem}{Theorem}[section]

\newtheorem{corollary}[theorem]{Corollary}

\newtheorem{proposition}[theorem]{Proposition}
\newtheorem{remark}[theorem]{Remark}

\makeatletter
\def\@seccntformat#1{\@ifundefined{#1@cntformat}%
	{\csname the#1\endcsname\quad}
	{\csname #1@cntformat\endcsname}
}
\makeatother

\newif\ifShowComments
\ShowCommentstrue
\def\strutdepth{\dp\strutbox}
\def\druk#1{\strut\vadjust{\kern-\strutdepth
        {\vtop to \strutdepth{%
                \baselineskip\strutdepth\vss
                        \llap{\hbox{#1}\quad}\null}}}}

\title{\bf
Unbiased estimation in new Gini index extensions under gamma distributions, with application to real income data
}

\author{
\text{Roberto Vila}$^{1}$\thanks{Corresponding author: Roberto Vila, email: {rovig161@gmail.com}
}
\,\,\,and
\text{Helton Saulo}$^{1,2}$
\\
{\small $^{1}$ Department of Statistics, University of Brasilia, Brasilia, Brazil}\\
{\small $^{2}$ Department of Economics, Federal University of Pelotas, Pelotas, Brazil}\\
}
\begin{document}
	\maketitle 	
	\begin{abstract}
%
{
In this paper, we introduce two flexible extensions of the classical Gini index, referred to as the extended lower and upper Gini indices. The proposed measures are based on the differences between an observation and the minimum and maximum order statistics in samples of size $m\geqslant 2$ and reduce to the classical Gini coefficient when $m=2$. Unlike conventional Gini-type measures, they provide a position-oriented assessment of inequality relative to the lower and upper tails of the distribution. We establish the consistency and asymptotic normality of the proposed estimators under mild regularity conditions. For gamma-distributed populations, we derive exact expressions for their expectations and prove their unbiasedness, thereby extending previous results of \cite{Deltas2003} and \cite{Baydil2025}. The finite-sample performance of the estimators is investigated through Monte Carlo simulations, and an application to 2023 GDP per capita data from South American countries illustrates the practical usefulness of the proposed measures. The results show that the extended lower and upper Gini indices provide a richer and more informative characterization of inequality than traditional Gini-type measures.
}
	\end{abstract}
	\smallskip
	\noindent
	{\small {\bfseries Keywords.} {Gamma distribution, extended lower Gini index (estimator), extended upper Gini index (estimator), $m$th Gini index (estimator), unbiased estimator.}}
	\\
	{\small{\bfseries Mathematics Subject Classification (2010).} {MSC 60E05 $\cdot$ MSC 62Exx $\cdot$ MSC 62Fxx.}}
%

\section{Introduction}

{
Measuring income inequality is a central issue in economics, with direct implications for welfare analysis, public policy, and economic development. The classical Gini index \citep{Gini1936} is one of the most widely used measures of inequality due to its simplicity and ease of interpretation. However, by condensing the entire income distribution into a single scalar, the Gini coefficient may obscure relevant distributional features. In particular, distinct income distributions can exhibit identical Gini values while differing substantially in terms of concentration at the lower, middle, or upper parts of the distribution.

Motivated by this limitation, recent studies have sought to enrich the informational content of Gini-type measures. Notably, \citet{Gavilan-Ruiz2024} introduced the $m$th Gini index, defined as the normalized expected range of random samples of size $m$. This index generalizes the classical Gini coefficient and allows inequality to be assessed across different sample sizes. Nevertheless, as it is solely based on the sample extremes, the $m$th Gini index remains an aggregate measure and does not provide insight into how inequality is distributed across specific positions within the income distribution.

In this paper, we introduce two novel Gini-type measures, referred to as the extended lower and extended upper Gini indices. These indices are constructed from the differences between a selected observation and, respectively, the smallest and largest observations in samples of size $m\geqslant 2$. Although the definitions involve a specific observation, the corresponding population coefficients are invariant with respect to its position. Indeed, by exchangeability,
\[
\mathbb{E}\left[X_i-\min\{X_1,\ldots,X_m\}\right]
=
\mathbb{E}\left[X_j-\min\{X_1,\ldots,X_m\}\right]
\]
and
\[
\mathbb{E}\left[\max\{X_1,\ldots,X_m\}-X_i\right]
=
\mathbb{E}\left[\max\{X_1,\ldots,X_m\}-X_j\right],
\]
for all $i,j=1,\ldots,m$. Consequently, the proposed population coefficients do not depend on the particular index chosen. In contrast, their sample estimators do depend on the selected position, since different observations within a subsample generally yield different finite-sample realizations. Thus, distinct estimators may be constructed for the same population coefficient, although all target the same underlying measure of inequality.

The proposed indices provide a position-oriented perspective on inequality. While the $m$th Gini index summarizes the expected range within a sample, the extended lower and upper Gini indices quantify how inequality is manifested relative to the lower and upper tails of the distribution. In particular, they generalize both the classical Gini coefficient and the $m$th Gini index and admit a natural decomposition whose sum coincides with the latter.

From an inferential perspective, we establish the consistency and asymptotic normality of the proposed estimators under mild regularity conditions. Furthermore, for gamma-distributed populations, we derive exact expressions for their expectations and prove their unbiasedness, thereby extending previous results available in the literature. The choice of the gamma distribution is motivated not only by its widespread use in modeling income and other nonnegative economic variables, but also by its analytical tractability. In particular, several expectations involving order statistics admit closed-form expressions under gamma models, making it possible to obtain exact finite-sample results. For many other positive-support distributions, analogous calculations become considerably more involved and often do not yield explicit formulas.

To the best of our knowledge, these are the first Gini-type measures that explicitly separate inequality contributions relative to the lower and upper sample extremes while possessing both exact finite-sample unbiasedness under gamma models and standard large-sample inferential properties. Monte Carlo simulations corroborate the theoretical findings, and an application to GDP per capita data illustrates how the proposed indices provide a richer and more informative assessment of inequality than traditional measures.

The remainder of the paper is organized as follows. Section~\ref{sec:02} introduces the extended lower and upper Gini indices and discusses their main properties and characterizations. Section~\ref{sec:03} presents the corresponding estimators and establishes their consistency, asymptotic normality, and unbiasedness. Section~\ref{sec:04} reports simulation results. Section~\ref{sec:05} applies the proposed indices to GDP per capita data, and Section~\ref{sec:06} concludes.
}

\section{New extended Gini indices, properties and characterizations}\label{sec:02}

Let $X_1, X_2, \ldots, X_m$ be independent and identically distributed (iid) random variables with the same
distribution as a non-negative random $X$ with mean $\mu=\mathbb{E}(X)>0$. 
For each integer $m\geqslant 2$, the {extended lower Gini index} of $X$ is defined as 
\begin{align}\label{extended-Gini}
	IG_{m;\text{min}}
	\equiv
	IG_{m;\text{min}}(X)
	=
	\dfrac{\mathbb{E}[X_{1}-\min\{X_1,\ldots,X_m\}]}{m\mu}.
\end{align}
Analogously, for each integer $m\geqslant 2$, the {extended upper Gini index} of $X$ is defined as 
\begin{align}\label{extended-Gini-1}
	IG_{m;\text{max}}
	\equiv
	IG_{m;\text{max}}(X)
	=
	\dfrac{\mathbb{E}[\max\{X_1,\ldots,X_m\}-X_{1}]}{m\mu}.
\end{align}

Note that $IG_{m;\text{min}}+IG_{m;\text{max}}$
reduces to $m$th Gini index, $IG_m$, recently introduced by  \cite{Gavilan-Ruiz2024}. Table~\ref{tab:extended_gini_comparison} highlights key differences between the \( m \)th Gini index of \cite{Gavilan-Ruiz2024} and the proposed extended indices. Note that while \( IG_m \) captures overall inequality via the sample range, the extended indices offer position-specific insights by comparing the first observation to the sample extremes.
%
%
\begin{table}[htb!]
	\centering
	\small
	\caption{Comparison between the $m$th Gini index and the proposed extended Gini indices.}
	\label{tab:extended_gini_comparison}
	\renewcommand{\arraystretch}{1.2}
	\resizebox{\textwidth}{!}{%
	\begin{tabular}{p{3.2cm} p{5.8cm} p{6.8cm}}
		\toprule
		\textbf{Aspect} &
		\textbf{$m$th Gini Index ($IG_m$)} &
		\textbf{Extended Gini Indices} \\
		\midrule
		
		Definition &
		$
		IG_m=
		\frac{\mathbb{E}[X_{m:m}-X_{1:m}]}{m\mu}
		$
		&
		$
		IG_{m;\mathrm{min}}
		=
		\frac{\mathbb{E}[X_1-X_{1:m}]}{m\mu},
		$
		
		$
		IG_{m;\mathrm{max}}
		=
		\frac{\mathbb{E}[X_{m:m}-X_1]}{m\mu}
		$
		\\
		
		Reference &
		Difference between the sample maximum and minimum &
		Difference between an observation and the sample minimum (lower) or maximum (upper)
		\\
		
		Population coefficients &
		One coefficient for each $m$ &
		Two coefficients for each $m$
		\\
		
		Estimator structure &
		Single estimator &
		A family of estimators indexed by $i$
		\\
		
		Decomposition &
		Aggregate measure &
		$
		IG_m
		=
		IG_{m;\mathrm{min}}
		+
		IG_{m;\mathrm{max}}
		$
		\\
		
		Interpretation &
		Overall inequality within a sample of size $m$ &
		Inequality relative to the lower and upper sample extremes
		\\
		
		\bottomrule
	\end{tabular}
}
\end{table}

Setting $m = 2$ such that $X_1 - \min\{X_1, X_2\}>0$ in \eqref{extended-Gini}, the extended lower Gini index recovers the standard Gini index \citep{Gini1936}.
\begin{align}\label{Gini coefficient}
	G\equiv IG_{2;\text{min}}= {\mathbb{E}\vert X_1-X_2\vert\over 2\mu}.
\end{align}
In a similar way, by taking $m =2$ such that $\max\{X_1,X_2\}-X_1>0$  in \eqref{extended-Gini-1}, the extended upper Gini index reduces to the standard Gini coefficient, that is, $G=IG_{2;\text{max}}$.

Propositions \ref{prop-1}-\ref{ext-gini-index-2} below provide fundamental properties and characterizations of the extended lower and upper Gini indices.

	\begin{proposition}\label{prop-1}
		For any $m\geqslant 2$, the expectations $\mathbb{E}[\min\{X_1,\ldots,X_m\}]$ and $\mathbb{E}[\max\{X_1,\ldots,X_m\}]$ exist. Consequently, the indices $IG_{m;\text{min}}$ and $IG_{m;\text{max}}$ defined in \eqref{extended-Gini} and \eqref{extended-Gini-1} are well defined.
	\end{proposition}
	\begin{proof}
	Since $F^m$ is the distribution function of $\max\{X_1,\ldots,X_m\}$, it follows for any $m\geqslant 2$ that
	\begin{align}\label{ineq-exp}
		\mathbb{E}[\min\{X_1,\ldots,X_m\}]\leqslant\mathbb{E}[\max\{X_1,\ldots,X_m\}]
		=
		\int_0^\infty x mF^{m-1}(x) {\rm d}F(x)
		\leqslant
		m
		\int_0^\infty x{\rm d}F(x)
		=
		m\mu.
	\end{align}
	The result follows immediately.
	\end{proof}
	
	\begin{proposition}\label{prop-init}
		For all $m\geqslant 2$ the extended Gini indices satisfy: $0\leqslant IG_{m;\text{min}}<1$ and $0\leqslant IG_{m;\text{max}}<1$.
	\end{proposition}
	\begin{proof}
The conclusion follows directly from definitions \eqref{extended-Gini}-\eqref{extended-Gini-1} and the bounds in \eqref{ineq-exp}.
	\end{proof}

	\begin{proposition}\label{prop-invariance}
		The extended Gini indices \eqref{extended-Gini} and \eqref{extended-Gini-1} satisfy:
		\begin{enumerate}
			\item \textit{Ratio-scale invariance:} For any $b>0$,
			\[
			IG_{m;\text{min}}(bX) = IG_{m;\text{min}}(X),
			\quad
			IG_{m;\text{max}}(bX) = IG_{m;\text{max}}(X),
			\]
			for all $m\geqslant 2$.
			
			\item \textit{Lack of translation invariance:} For any $a>0$,
			\[
			IG_{m;\text{min}}(a+X)
			=
			\frac{\mu}{a+\mu}\,IG_{m;\text{min}}(X),
			\quad
			IG_{m;\text{max}}(a+X)
			=
			\frac{\mu}{a+\mu}\, IG_{m;\text{max}}(X),
			\]
			for all $m\geqslant 2$,
			where $\mu=\mathbb{E}[X]$.
		\end{enumerate}
	\end{proposition}
	\begin{proof}
The proof is an immediate consequence of definitions \eqref{extended-Gini} and \eqref{extended-Gini-1}, and the proof is omitted.
	\end{proof}

\begin{proposition}\label{prop-using}
	Let $m\geqslant 2$. The extended Gini indices defined in
	\eqref{extended-Gini} and \eqref{extended-Gini-1} admit the following
	covariance representations:
	\[
	IG_{m;\text{min}}
	=
	\frac{1}{\mu}\,
	\mathrm{Cov}\!\left(
	X,\,
	1-(1-F(X))^{m-1}
	\right),
	\quad 
	IG_{m;\text{max}}
	=
	\frac{1}{\mu}\,
	\mathrm{Cov}\!\left(
	X,\,
	F^{m-1}(X)-1
	\right),
	\]
	where $F$ denotes the distribution function of $X$ and
	$\mu=\mathbb{E}[X]$.
\end{proposition}
\begin{proof}
	Since $\mathbb{E}\!\left[(1-F(X))^{m-1}\right]=1/m$, it follows that
	\[
	\frac{1}{\mu}\,
	\mathrm{Cov}\!\left(
	X,\,1-(1-F(X))^{m-1}
	\right)
	=
	\frac{1}{m\mu}\,
	\mathbb{E}\!\left[
	F^{-1}(U)\{1-m(1-U)^{m-1}\}
	\right],
	\]
	where $U\sim {U}(0,1)$.  
	The first identity then follows by combining this expression with
	\[
	\mathbb{E}[\min\{X_1,\ldots,X_m\}]
	=
	m\,\mathbb{E}[F^{-1}(U)(1-U)^{m-1}]
	\]
	and the definition \eqref{extended-Gini}.
	
	Similarly, using $\mathbb{E}[F^{m-1}(X)]=1/m$, we obtain
	\[
	\frac{1}{\mu}\,
	\mathrm{Cov}\!\left(
	X,\,F^{m-1}(X)-1
	\right)
	=
	\frac{1}{m\mu}\,
	\mathbb{E}\!\left[
	F^{-1}(U)\{mU^{m-1}-1\}
	\right],
	\]
	with $U\sim {U}(0,1)$.  
	The second identity follows from the relation
	\[
	\mathbb{E}[\max\{X_1,\ldots,X_m\}]
	=
	m\,\mathbb{E}[F^{-1}(U)U^{m-1}]
	\]
	together with the definition \eqref{extended-Gini-1}.
\end{proof}

\begin{remark}
	Summing the lower and upper components yields the $m$th Gini index introduced by \citet{Gavilan-Ruiz2024}:
	\[
	IG_m
	=
	IG_{m;\text{min}}+IG_{m;\text{max}}
	=
	\frac{1}{\mu}\,
	\mathrm{Cov}\!\left(
	X,\,
	F_{X_{m-1:m-1}}(X)+F_{X_{1:m-1}}(X)
	\right),
	\]
	where $F_{X_{m-1:m-1}}(x)=F^{m-1}(x)$ and $F_{X_{1:m-1}}(x)=1-[1-F(x)]^{m-1}$ denote the distribution functions of the order statistics $X_{m-1:m-1}$ and $X_{1:m-1}$, respectively.
	
	For $m=2$, this expression reduces to the classical Gini coefficient \eqref{Gini coefficient} \cite[see][]{Yin2024}:
	\[
	G
	=
	IG_{2;\text{min}}+IG_{2;\text{max}}
	=
{1\over\mu}
	\mathrm{Cov}(X,2F(X)-1)
	=
	{1\over\mu}
	\int_0^1 F^{-1}(p)(2p-1){\rm d}p
	=
	\frac{R_G(F)}{\mu},
	\]
	where
	$
	R_G(F)
	\equiv
	\mathbb{E}|X_1-X_2|/2
	$
	is the Gini mean difference.
\end{remark}

	The result below establishes a connection between the extended Gini indices \eqref{extended-Gini} and \eqref{extended-Gini-1} and Lorenz measures of inequality.	
	\begin{proposition}\label{charact}
		Let $m\geqslant 2$. 
	The extended Gini indices \eqref{extended-Gini} and \eqref{extended-Gini-1} can be written as
		\begin{align*}
IG_{m;\text{min}}
=
{1\over m}\, G_m(F),
		\quad
		IG_{m;\text{max}}
=
\left(1-{1\over m}\right)
D_{m-1}(F),
		\end{align*}
		where 
		$L(p)={\int_0^p F^{-1}(t){\rm d}t/\mu}$, for any $0\leqslant p\leqslant 1$,
is the Lorenz curve for $X$, 
		\begin{align*}
			D_n\equiv D_n(F)=(n+1)\mathbb{E}[\{U-L(U)\}U^{n-1}], 
			\quad U\sim U(0,1), \quad n\geqslant 1,
		\end{align*}
		is the Lorenz measure of inequality introduced in \cite{Aaberge2000}, and
				\begin{align*}
			G_n
			\equiv
			G_n(F)
			=
			n(n-1)\mathbb{E}[\{U-L(U)\}(1-U)^{n-2}], 
			\quad U\sim U(0,1), \quad n\geqslant 1,
		\end{align*}
		is the generalized Gini measure introduced in \cite{Kakwani1980,Donaldson1980,Yitzhaki1983}.
	\end{proposition}
		\begin{proof}
Since $L'(p)=F^{-1}(p)/\mu$ for $0<p<1$, it follows from Proposition
\ref{prop-using} that
\[
IG_{m;\text{min}}
=
\frac{1}{m}\,
\mathbb{E}\!\left[
\{L'(U)-1\}\{1-m(1-U)^{m-1}\}
\right],
\quad
IG_{m;\text{max}}
=
\frac{1}{m}\,
\mathbb{E}\!\left[
\{L'(U)-1\}\{mU^{m-1}-1\}
\right],
\]
where $U\sim {U}(0,1)$.  
Applying integration by parts, these expressions can be rewritten as
\[
IG_{m;\text{min}}
=
\frac{1}{m}\,
\mathbb{E}\!\left[
\{U-L(U)\}\, m(m-1)(1-U)^{m-2}
\right],
\quad 
IG_{m;\text{max}}
=
\frac{1}{m}\,
\mathbb{E}\!\left[
\{U-L(U)\}\, m(m-1)U^{m-2}
\right],
\]
respectively.

Finally, the desired result follows by using the definitions of
$D_n(F)$ and $G_n(F)$.
	\end{proof}
	
	\begin{remark}		
Summing the lower and upper components yields
		\begin{align*}
			IG_m
			=
			IG_{m;\text{min}}+IG_{m;\text{max}}
			=
			\left(1-{1\over m}\right) 
			D_{m-1}(F)
			+
			{1\over m}\, G_m(F),
			\quad 
			{m\geqslant 2}.
		\end{align*}
The above expression for $IG_m$ has previously appeared in \cite{Gavilan-Ruiz2024}.
	\end{remark}	
	
	\begin{proposition}\label{pro-eq}
Let $X$ be a non-negative random variable with mean $\mu>0$. For any integer $m\geqslant 2$  there exist constants $r_m,s_m>0$ such that
\[
IG_{m;\text{min}}(X)=G(X+r_m),
\quad
IG_{m;\text{max}}(X)=G(X+s_m),
\]
where $G$ denotes the classical Gini coefficient \eqref{Gini coefficient}.
	\end{proposition}
	\begin{proof}
Let $t \equiv IG_{m;\text{min}}(X)$ and define $f(r)=G(X+r)$ for $r\geqslant 0$. Since $f$ is continuous and strictly decreasing, with
\[
f(\infty)=0<t\leqslant IG_m(X)\leqslant G(X)=f(0^+),
\]
the intermediate value theorem ensures the existence of $r_m\in(0,\infty)$ such that $f(r_m)=t$.
The argument for the remaining identity is analogous, which completes the proof.
	\end{proof}

	
	\begin{remark}
Proposition \ref{pro-eq} shows that the indices \eqref{extended-Gini}–\eqref{extended-Gini-1} are themselves genuine Gini coefficients. Moreover, the extended lower Gini index $IG_{m;\text{min}}$ (and similarly $IG_{m;\text{max}}$) can be represented as the classical Gini coefficient applied to the shifted variable $X+r_m$, with $r_m$ determining the magnitude of the shift.
	\end{remark}

\begin{proposition}\label{ext-gini-index-0}
The extended lower Gini index \eqref{extended-Gini} can be written as
{
\begin{align*}
	IG_{m;\text{min}}
	=
	\dfrac{		\displaystyle 
\int_0^\infty 
\left[1-\mathbb{P}\left(X\leqslant t\right)\right]
{\rm d}t
		-	
	\int_0^\infty 
\{1-\mathbb{P}\left(X\leqslant t\right)\}^m
{\rm d}t
		}{
		\displaystyle
		m\int_0^\infty 
		\left[1-\mathbb{P}\left(X\leqslant t\right)\right]
		{\rm d}t
		}.
\end{align*}
}
\end{proposition}
\begin{proof}
Using the identity
\begin{align}\label{id-minimum-0}
	\min\{X_1,\ldots,X_m\}
	=
	\int_0^\infty 
	\mathds{1}_{\bigcap_
		{i=1}^m\{X_i\geqslant t\}}
	{\rm d}t,
\end{align}
we have
\begin{align}\label{exp-min}
	\mathbb{E}[X_{1}-\min\{X_1,\ldots,X_m\}]
	&=
	\mu-
	\int_0^\infty 
	\mathbb{P}\left(\bigcap_
	{i=1}^m\{X_i\geqslant t\}\right)
	{\rm d}t
	=
		\mu-
	\int_0^\infty 
	\{1-\mathbb{P}\left(X\leqslant t\right)\}^m
	{\rm d}t,
\end{align}
where Tonelli's Theorem permits the change in integration order, and the final step follows from the independence and identical distribution nature of $X_1, X_2, \ldots, X_m$.

Combining \eqref{exp-min}, the well-known identity  $\mu = \int_0^\infty [1 - F(t)] \, {\rm d}t$, and the extended upper Gini index definition \eqref{extended-Gini-1} yields the result.
\end{proof}

\begin{proposition}\label{ext-gini-index-1}
	The extended upper Gini index \eqref{extended-Gini-1} takes the form:
		\begin{align*}
			IG_{m;\text{max}}
			=
			\dfrac{\displaystyle 
				\int_0^\infty 
				\left[1-\{\mathbb{P}\left(X\leqslant t\right)\}^m\right]
				{\rm d}t
				-	
				\int_0^\infty 
				\left[1-\mathbb{P}\left(X\leqslant t\right)\right]
				}{
				\displaystyle
				m\int_0^\infty 
				\left[1-\mathbb{P}\left(X\leqslant t\right)\right]
				{\rm d}t}.
		\end{align*}
\end{proposition}
\begin{proof}
By applying the identity
\begin{align}\label{id-maximum}
	\max\{X_1,\ldots,X_m\}
	=
	\int_0^\infty 
	\left[1-\mathds{1}_{\bigcap_
		{i=1}^m\{X_i\leqslant t\}}\right]
	{\rm d}t,
\end{align}
we obtain
\begin{align}\label{exp-min-0}
	\mathbb{E}[\max\{X_1,\ldots,X_m\}-X_{i}]
	&=
	\int_0^\infty 
	\left[1-\mathbb{P}\left(\bigcap_
		{i=1}^m\{X_i\leqslant t\}\right)\right]
		{\rm d}t
		-	
	\mu
		=
		\int_0^\infty 
		\left[1-\{\mathbb{P}\left(X\leqslant t\right)\}^m\right]
		{\rm d}t
		-	
		\mu,
\end{align}
where Tonelli's Theorem justifies interchanging the integration order, and with the final step resulting from the independence and identical distribution of  $X_1, X_2, \ldots, X_m$.

Then, the result follows by combining \eqref{exp-min-0}, the well-known identity  $\mu = \int_0^\infty [1 - F(t)] \, {\rm d}t$ and the extended upper Gini index definition \eqref{extended-Gini-1}.
\end{proof}

\begin{proposition}\label{ext-gini-index}
	The extended lower Gini index for $X\sim \text{Gamma}(\alpha,\lambda)$ (gamma distribution) is given by
	\begin{align}\label{mthgini_gamma}
		IG_{m;\text{min}}
		=
		{1\over m}
		\left[
		1
			-	
					{1\over\alpha}
			\int_0^\infty 
			\left\{1-{\gamma(\alpha, t)\over\Gamma(\alpha)}\right\}^m
			{\rm d}t
		\right].
	\end{align}
\end{proposition}
\begin{proof}
The result follows directly from the Proposition \ref{ext-gini-index-0} and is thus omitted.
\end{proof}

\begin{proposition}\label{ext-gini-index-2}
	The extended upper Gini index for $X\sim \text{Gamma}(\alpha,\lambda)$ is given by
	\begin{align}\label{mthgini_gamma-2}
		IG_{m;\text{max}}
		&=
		{1\over m}
		\left[
		{1\over\alpha}
		\int_0^\infty 
		\left\{1-{\gamma^m(\alpha, t)\over\Gamma^m(\alpha) }\right\}
		{\rm d}t
		-
				1
		\right].
	\end{align}
\end{proposition}
\begin{proof}
The proof follows immediately from Proposition \ref{ext-gini-index-1} and is omitted for brevity
\end{proof}

\begin{remark}
	Except for $m = 2$ \citep[see Remark 2.6 of][]{Vila2025}, the integrals in Propositions \ref{ext-gini-index} and \ref{ext-gini-index-2} lack closed-form expressions in terms of standard mathematical functions, necessitating numerical integration methods
\end{remark}

\section{Unbiasedness of extended Gini index estimators}\label{sec:03}

{
This section focuses on deriving explicit expressions for the expected values of the extended lower and upper Gini index estimators,
$\widehat{\,_iIG}_{m;\mathrm{min}}$ and
$\widehat{\,^iIG}_{m;\mathrm{max}}$,
for $i=1,\ldots,m$. In addition, we establish their consistency and asymptotic normality under mild regularity conditions. The estimators are defined as follows:
}
\begin{align}\label{estimator}
	\widehat{\,_iIG}_{m;\text{min}}
	=
	{(m-1)!\over (n-1)(n-2)\cdots(n-m+1)} \,
	\dfrac{\displaystyle
		\sum_{1\leqslant j_1<\cdots< j_m\leqslant n}
		\left[
		X_{j_i}
		-
		\min\{X_{j_1},\ldots,X_{j_m}\}
		\right]
	}{\displaystyle \sum_{k=1}^{n}X_k}
\end{align}
and
\begin{align}\label{estimator-1}
	\widehat{\,^iIG}_{m;\text{max}}
	=
	{(m-1)!\over (n-1)(n-2)\cdots(n-m+1)} \,
	\dfrac{\displaystyle
		\sum_{1\leqslant j_1<\cdots< j_m\leqslant n}
		\left[
		\max\{X_{j_1},\ldots,X_{j_m}\}
		-
		X_{j_i}
		\right]
	}{\displaystyle \sum_{k=1}^{n}X_k},
\end{align}
respectively,
where $X_{i_1}, X_{i_2},\ldots, X_{i_m}$ are
iid observations  of $X$. 

{
\begin{remark}
	The population coefficients $IG_{m;\mathrm{min}}$ and $IG_{m;\mathrm{max}}$, defined in \eqref{extended-Gini} and \eqref{extended-Gini-1}, respectively, are independent of the particular choice of $X_1$ in the numerator. Indeed, by exchangeability,
	\[
	\mathbb{E}\left[X_1-\min\{X_1,\ldots,X_m\}\right]
	=
	\mathbb{E}\left[X_i-\min\{X_1,\ldots,X_m\}\right],
	\quad i=1,\ldots,m,
	\]
	and
	\[
	\mathbb{E}\left[\max\{X_1,\ldots,X_m\}-X_1\right]
	=
	\mathbb{E}\left[\max\{X_1,\ldots,X_m\}-X_i\right],
	\quad i=1,\ldots,m.
	\]
	In contrast, the estimators $\widehat{\,_iIG}_{m;\mathrm{min}}$ and $\widehat{\,^iIG}_{m;\mathrm{max}}$ depend on $i$, so different choices of $i$ generally produce different finite-sample estimates of the same population coefficients.
\end{remark}
}

\begin{remark}\label{rem-gini-index}
	Note that
	\begin{align*}
			\widehat{IG}_m
			&\equiv
	\widehat{\,_iIG}_{m;\text{min}}+\widehat{\,^iIG}_{m;\text{max}}
	\\[0,2cm]
	&=
		{(m-1)!\over (n-1)(n-2)\cdots(n-m+1)} \,
	\dfrac{\displaystyle
		\sum_{1\leqslant j_1<\cdots< j_m\leqslant n}
		\left[
		\max\{X_{j_1},\ldots,X_{j_m}\}
		-
		\min\{X_{j_1},\ldots,X_{j_m}\}
		\right]
	}{\displaystyle \sum_{i=1}^{n}X_i},
	\end{align*}
	where $	\widehat{IG}_m$ is the $m$th Gini index estimator proposed in \cite{Vila2025}.
\end{remark}

{
\begin{proposition}[Consistency]
	Let $X_1,X_2,\ldots$ be i.i.d. random variables with mean $\mu>0$.
	Then, for each fixed $m\geqslant 2$ and $i=1,\ldots,m$,
	\[
	\widehat{\,_iIG}_{m;\mathrm{min}}
	\overset{\mathrm{a.s.}}{\longrightarrow}
	IG_{m;\mathrm{min}}
	\quad\text{and}\quad
	\widehat{\,^iIG}_{m;\mathrm{max}}
	\overset{\mathrm{a.s.}}{\longrightarrow}
	IG_{m;\mathrm{max}},
	\]
	as $n\to\infty$, where $\overset{\mathrm{a.s.}}{\longrightarrow}$ denotes almost sure convergence.
\end{proposition}
\begin{proof}
	Note that the numerators of
	$\widehat{\,_iIG}_{m;\mathrm{min}}$
	and
	$\widehat{\,^iIG}_{m;\mathrm{max}}$
	are U-statistics of degree $m$ with kernels
	\[
	h_i(x_1,\ldots,x_m)
	\equiv 
	x_i-\min\{x_1,\ldots,x_m\}
	\quad\text{and}\quad
	g_i(x_1,\ldots,x_m)
	\equiv
	\max\{x_1,\ldots,x_m\}-x_i,
	\]
	respectively. By the strong law of large numbers for U-statistics
	\citep[Theorem~3.1.1]{Lee1990},
	their numerators converge almost surely to the corresponding expectations. Moreover,
	$
	(1/n)\sum_{k=1}^{n}X_k
	\overset{\mathrm{a.s.}}{\longrightarrow}
	\mu
	$
	by the classical strong law of large numbers. The result follows from the continuous mapping theorem.
\end{proof}

\begin{proposition}[Asymptotic normality]
	Assume that $\mathbb E[X^2]<\infty$.
	Then, for each fixed $m\geqslant 2$ and $i=1,\ldots,m$,
	\[
	\sqrt n
	\big(
	\widehat{\,_iIG}_{m;\mathrm{min}}
	-
	IG_{m;\mathrm{min}}
	\big)
	\overset{d}{\longrightarrow}
	N(0,\sigma_{i,\mathrm{min}}^2)
	\quad 
	\text{and}
	\quad
	\sqrt n
	\big(
	\widehat{\,^iIG}_{m;\mathrm{max}}
	-
	IG_{m;\mathrm{max}}
	\big)
	\overset{d}{\longrightarrow}
	N(0,\sigma_{i,\mathrm{max}}^2),
	\]
where $\overset{d}{\longrightarrow}$ denotes convergence in distribution, and
$\sigma_{i,\mathrm{min}}^2$ and
$\sigma_{i,\mathrm{max}}^2$ are finite asymptotic variances that can be derived using the Hoeffding decomposition and the multivariate delta method \citep{Hoeffding1948,Lee1990}.
\end{proposition}
\begin{proof}
	Since $\mathbb E[X^2]<\infty$, the kernels
	$h_i$ and $g_i$ have finite second moments. Therefore, by the classical asymptotic normality theorem for U-statistics
	\citep{Hoeffding1948,Lee1990},
	the numerators of
	$\widehat{\,_iIG}_{m;\mathrm{min}}$
	and
	$\widehat{\,^iIG}_{m;\mathrm{max}}$
	are asymptotically normal. Furthermore,
	$
	(1/n)\sum_{k=1}^{n}X_k
	\overset{p}{\longrightarrow}
	\mu,
	$
	and hence the stated limits follow from Slutsky's theorem.
\end{proof}
}

\begin{theorem}\label{main-theorem}
Let $X_1, X_2, \ldots, X_m$ be independent copies of a non-negative and absolutely continuous random variable $X$ with finite and positive expected value and common cumulative distribution function $F$. For each $i=1,\ldots,m$, the following statements hold:
\begin{align*}
\mathbb{E}[\widehat{\,_iIG}_{m;\text{min}}]
	&=\!
	{n\over m}\!
\left[
\int_0^\infty\!
\mathbb{E}\left[
X \exp\left(-Xz\right)
\right]
\mathscr{L}_F^{n-1}(z)
{\rm d}z\!
-\!
\int_0^\infty\!
\int_0^\infty\!
\mathbb{E}^m\left[\!
\mathds{1}_{\{X\geqslant t\}}
\exp\left(-X z\right)\!
\right]
{\rm d}t
\mathscr{L}_F^{n-m}(z)
{\rm d}z
\right],
\end{align*}
and
\begin{multline*}
	\mathbb{E}[\widehat{\,^iIG}_{m;\text{max}}]
	=
	{n\over m}
	\Bigg[	
	\int_0^\infty
	\int_0^\infty
	\left\{
	\mathscr{L}_F^{m}(z)
		-
	\mathbb{E}^m\left[
	\mathds{1}_{\{X\leqslant t\}}
	\exp\left(-X z\right)
	\right]
	\right\}
	{\rm d}t
	\mathscr{L}_F^{n-m}(z)
	{\rm d}z
	\\[0,2cm]
	-
	\int_0^\infty
	\mathbb{E}\left[
	X \exp\left(-Xz\right)
	\right]
	\mathscr{L}_F^{n-1}(z)
	{\rm d}z
		\Bigg],
\end{multline*}
where $\mathscr{L}_F(z)=\int_0^\infty \exp(-zx){\rm d}F(x)$ is the Laplace transform associated with distribution $F$. In the above, we are assuming that the expectations and improper integrals converge.
\end{theorem}
\begin{proof}
By using the identity
$\int_{0}^\infty \exp(-w z){\rm d}z={1/ w}$, $ w>0, $
with $w=\sum_{k=1}^{n}X_k$, we get
\begin{multline}\label{exp-1}
	\mathbb{E}\left[
	\dfrac{\displaystyle
		\sum_{1\leqslant j_1<\cdots< j_m\leqslant n}
\left[
X_{j_i}
-
\min\{X_{j_1},\ldots,X_{j_m}\}
\right]
	}{\displaystyle \sum_{k=1}^{n}X_k}
	\right]
	=
	\sum_{1\leqslant j_1<\cdots< j_m\leqslant n}
	\mathbb{E}\left[X_{j_i}
	\int_0^\infty\exp\left\{-\left(\sum_{k=1}^{n}X_k\right)z\right\}{\rm d}z\right]
	\\[0,2cm]
	-
	\sum_{1\leqslant j_1<\cdots< j_m\leqslant n}
	\mathbb{E}\left[
	\min\{X_{j_1},\ldots,X_{j_m}\}
	\int_0^\infty\exp\left\{-\left(\sum_{k=1}^{n}X_k\right)z\right\}{\rm d}z\right]
		\\[0,2cm]
		=
	\sum_{1\leqslant j_1<\cdots< j_m\leqslant n} \
\int_0^\infty
\mathbb{E}\left[X_{j_i}
\exp\left\{-\left(\sum_{k=1}^{n}X_k\right)z\right\}
\right]
{\rm d}z
\\[0,2cm]
-
\sum_{1\leqslant j_1<\cdots< j_m\leqslant n} \
\int_0^\infty
\mathbb{E}\left[
\min\{X_{j_1},\ldots,X_{j_m}\}
\exp\left\{-\left(\sum_{k=1}^{n}X_k\right)z\right\}
\right]
{\rm d}z,
\end{multline}
where Tonelli's Theorem justifies the interchange of integrals.
Utilizing the identity \eqref{id-minimum-0},
the above expression \eqref{exp-1} becomes
\begin{align}\label{exp-2}
&=
\sum_{1\leqslant j_1<\cdots<j_m\leqslant n} \
\int_0^\infty
\mathbb{E}\left[
X_{j_i}
\exp\left(-X_1 z\right)
\exp\left\{-\left(\sum_{k=2}^{n}X_k\right)z\right\}
\right]
{\rm d}z
\nonumber
\\[0,2cm]
&-
\sum_{1\leqslant j_1<\cdots<j_m\leqslant n} \
\int_0^\infty
\mathbb{E}\left[
	\int_0^\infty 
\mathds{1}_{\bigcap_
	{k=1}^m\{X_{j_k}\geqslant t\}}
{\rm d}t
\exp\left\{-\left(\sum_{k=1}^{m}X_k\right)z\right\}
\exp\left\{-\left(\sum_{k=m+1}^{n}X_k\right)z\right\}
\right]
{\rm d}z
\nonumber
\\[0,2cm]
&=
\sum_{1\leqslant j_1<\cdots<j_m\leqslant n} \
\int_0^\infty
\mathbb{E}\left[
X_{j_i}
\exp\left(-X_1 z\right)
\exp\left\{-\left(\sum_{k=2}^{n}X_k\right)z\right\}
\right]
{\rm d}z
\nonumber
\\[0,2cm]
&-
\sum_{1\leqslant j_1<\cdots<j_m\leqslant n} \
\int_0^\infty
\int_0^\infty 
\mathbb{E}\left[
\mathds{1}_{\bigcap_
	{k=1}^m\{X_{j_k}\geqslant t\}}
\exp\left\{-\left(\sum_{k=1}^{m}X_k\right)z\right\}
\exp\left\{-\left(\sum_{k=m+1}^{n}X_k\right)z\right\}
\right]
{\rm d}t
{\rm d}z,
\end{align}
where Tonelli's Theorem again justifies the interchange of integrals.
As $X_1, X_2, \ldots, X_m$ are iid, the expression  \eqref{exp-2} simplifies to
\begin{align*}
	&=	
\sum_{1\leqslant j_1<\cdots<j_m\leqslant n}
	\int_0^\infty
	\mathbb{E}\left[
	X_{j_i}
	\exp\left(-X_1 z\right)
	\mathbb{E}\left[
	\exp\left\{-\left(\sum_{k=2}^{n}X_k\right)z\right\}
	\right]
	\right]
	{\rm d}z
	\\[0,2cm]
	&-
\sum_{1\leqslant j_1<\cdots<j_m\leqslant n}
	\int_0^\infty
		\int_0^\infty
	\mathbb{E}\left[
\mathds{1}_{\bigcap_
	{k=1}^m\{X_{j_k}\geqslant t\}}
	\exp\left\{-\left(\sum_{k=1}^{m}X_k\right)z\right\}
	\mathbb{E}
	\left[
	\exp\left\{-\left(\sum_{k=m+1}^{n}X_k\right)z\right\}
	\right]
	\right]
		{\rm d}t
	{\rm d}z
		\\[0,2cm]
	&=
\binom{n}{m}
\int_0^\infty
\mathbb{E}\left[
X \exp\left(-Xz\right)
\right]
\mathscr{L}_F^{n-1}(z)
{\rm d}z
-
\binom{n}{m}
\int_0^\infty
\int_0^\infty
\mathbb{E}^m\left[
\mathds{1}_{\{X\geqslant t\}}
\exp\left(-X z\right)
\right]
{\rm d}t
\mathscr{L}_F^{n-m}(z)
{\rm d}z.
\end{align*}
This concludes the proof of the identity for $\mathbb{E}[\widehat{\,_iIG}_{m;\text{min}}]$.

The derivation of $\mathbb{E}[\widehat{\,^iIG}_{m;\text{max}}]$ is analogous to that of $\mathbb{E}[\widehat{\,_iIG}_{m;\text{min}}]$, using identity \eqref{exp-min-0} instead of  \eqref{id-minimum-0}, so the proof is omitted for brevity.

Thus, we have complete the proof.
\end{proof}

We now apply Theorem \ref{main-theorem}  to get explicit formulas of $\mathbb{E}[\widehat{\,_iIG}_{m;\text{min}}]$ and $\mathbb{E}[\widehat{\,^iIG}_{\text{max}}]$ in gamma populations, confirming their unbiasedness.

\begin{corollary}\label{corollary-main}
	Let $X_1, X_2,\ldots, X_m$ be independent copies of $X\sim\text{Gamma}(\alpha,\lambda)$. For each $i=1,\ldots,m$, we have:
	\begin{align*}
	\mathbb{E}[\widehat{\,_iIG}_{m;\text{min}}]
&=
{1\over m}
\left[
1
-
{1\over\alpha}
\int_0^\infty
\left\{1-
{\gamma(\alpha,v)\over \Gamma(\alpha)} 
\right\}^m
{\rm d}v 
\right]
=
IG_{m;\text{min}},
\end{align*}
where $IG_{m;\text{min}}$ is the extended lower Gini index given in Proposition \ref{ext-gini-index}. Thus, the estimator $\widehat{\,_iIG}_{m;\text{min}}$ is unbiased
for gamma populations.	
\end{corollary}
\begin{proof}
For $X\sim\text{Gamma}(\alpha,\lambda)$ direct computation gives
\begin{align}\label{eq-1}
\mathbb{E}\left[
X \exp\left(-Xz\right)
\right]
=
{\alpha\lambda^\alpha\over(z+\lambda)^{\alpha+1}},
\quad
\mathbb{E}\left[
\mathds{1}_{\{X\leqslant t\}}
\exp\left(-X z\right)
\right]
=
{\lambda^\alpha\over (z+\lambda)^\alpha} \, {\gamma(\alpha,(z+\lambda)t)\over \Gamma(\alpha)}.
\end{align}
As $\mathscr{L}_F(z)=\lambda^\alpha/(z+\lambda)^\alpha$, Theorem \ref{main-theorem} yields
\begin{align*}
	\mathbb{E}[\widehat{\,_iIG}_{m;\text{min}}]
	&=
	{n\over m}
	\left[
	\int_0^\infty
{\alpha \lambda^{\alpha n}\over (z+\lambda)^{\alpha n+1}}
	{\rm d}z
	-
	\int_0^\infty
	\int_0^\infty
	\left\{
	1-
{\gamma(\alpha,(z+\lambda)t)\over \Gamma(\alpha)} 
\right\}^m
	{\rm d}t \,
{\lambda^{\alpha n}\over (z+\lambda)^{\alpha n}}
	{\rm d}z
	\right].
\end{align*}
Making the change of variable $v= (z + \lambda)t$, the above identity becomes
\begin{align*}
	\mathbb{E}[\widehat{\,_iIG}_{m;\text{min}}]
	&=
	{1\over m}
	\left[
1
	-
	{1\over\alpha}
	\int_0^\infty
	\left\{1-
	{\gamma(\alpha,v)\over \Gamma(\alpha)} 
	\right\}^m
	{\rm d}v 
	\right]
	\int_0^\infty
	{\alpha n \lambda^{\alpha n}\over (z+\lambda)^{\alpha n+1}}
	{\rm d}z.
\end{align*}
Since $	\int_0^\infty
{\alpha n \lambda^{\alpha n}/ (z+\lambda)^{\alpha n+1}}
{\rm d}z=1$, from Proposition \ref{ext-gini-index} the proof follows.
\end{proof}

\begin{corollary}\label{corollary-main-1}
	Let $X_1, X_2,\ldots, X_m$ be independent copies of $X\sim\text{Gamma}(\alpha,\lambda)$. For each $i=1,\ldots,m$, we have:
	\begin{align*}
		\mathbb{E}[\widehat{\,^iIG}_{m;\text{max}}]
		&=
		{1\over m}
\left[
{1\over\alpha}
\int_0^\infty 
\left\{1-{\gamma^m(\alpha, v)\over\Gamma^m(\alpha) }\right\}
{\rm d}v
-
1
\right]
		=
	IG_{m;\text{max}},
	\end{align*}
	where $IG_{m;\text{max}}$ is the extended upper Gini index given in Proposition \ref{ext-gini-index-2}. Hence, the estimator $\widehat{\,^iIG}_{m;\text{max}}$ is unbiased
	for gamma populations.	
\end{corollary}
\begin{proof}
	As $\mathscr{L}_F(z)=\lambda^\alpha/(z+\lambda)^\alpha$, by using \eqref{eq-1} in Theorem \ref{main-theorem}, we get
\begin{align*}
	\mathbb{E}[\widehat{\,^iIG}_{m;\text{max}}]
	=
	{n\over m}
	\Bigg[	
	\int_0^\infty
	\int_0^\infty
	\left\{
	1
	-
	{\gamma^m(\alpha,(z+\lambda)t)\over \Gamma^m(\alpha)}
	\right\}
	{\rm d}t \,
{\lambda^{\alpha n}\over (z+\lambda)^{\alpha n}}
	{\rm d}z
	-
	\int_0^\infty
{\alpha\lambda^{\alpha n}\over (z+\lambda)^{\alpha n+1}}
	{\rm d}z
	\Bigg].
\end{align*}
With the change of variable $v= (z + \lambda)t$, the above identity transforms into
\begin{align*}
	\mathbb{E}[\widehat{\,^iIG}_{m;\text{max}}]
	=
	{n\over m}
	\Bigg[	
	{1\over\alpha}
	\int_0^\infty
	\left\{
	1
	-
	{\gamma^m(\alpha,v)\over \Gamma^m(\alpha)}
	\right\}
	{\rm d}v 
	-
1
	\Bigg]
		\int_0^\infty
	{\alpha\lambda^{\alpha n}\over (z+\lambda)^{\alpha n+1}}
	{\rm d}z.
\end{align*}
Given that $\int_0^\infty {\alpha n \lambda^{\alpha n}}/{(z+\lambda)^{\alpha n+1}} {\rm d}z = 1$, Proposition \ref{ext-gini-index-2} yields the result.
\end{proof}

\begin{remark}
	Note that Corollaries \ref{corollary-main} and \ref{corollary-main-1} generalizes prior findings by \cite{Deltas2003} and \cite{Baydil2025}.
\end{remark}

\begin{remark}
The scale invariance of $\widehat{\,_iIG}_{m;\text{min}}$ and $\widehat{\,^iIG}_{m;\text{max}}$ implies that $\mathbb{E}[\widehat{\,_iIG}_{m;\text{min}}]$ and $\mathbb{E}[\widehat{\,^iIG}_{m;\text{max}}]$ do not depend on $\lambda$, as stated in Corollaries \ref{corollary-main} and \ref{corollary-main-1}.
\end{remark}

Corollaries \ref{corollary-main} and \ref{corollary-main-1} imply the following result.
\begin{proposition}\label{prop-main}
	Let $X_1, X_2,\ldots, X_m$ be independent copies of $X\sim\text{Gamma}(\alpha,\lambda)$. For each $i=1,\ldots,m$, we have:
		\begin{align*}
		\mathbb{E}[\widehat{IG}_m]
		\equiv
		\mathbb{E}[\widehat{\,_iIG}_{m;\text{min}}]
		+
		\mathbb{E}[\widehat{\,^iIG}_{m;\text{max}}]
		&=
		{1\over \alpha m}
		\left[
		\int_0^\infty 
		\left\{
		1
		-
		{\gamma^m(\alpha,v)\over\Gamma^m(\alpha)}
		\right\}
		{\rm d}v  \,
		\nonumber
		-
		\int_0^\infty 
		\left\{
		1
		-
		{\gamma(\alpha,v)\over\Gamma(\alpha)}
		\right\}^m
		{\rm d}v
		\right]
		\\[0,2cm]
		&=
		{IG}_{m;\text{min}}
		+
		{IG}_{m;\text{max}}
				\\[0,2cm]
		&
		\equiv
		IG_m
		= 
		\dfrac{\mathbb{E}[\max\{X_1,\ldots,X_m\}-\min\{X_1,\ldots,X_m\}]}{m\mu},
	\end{align*}
	where $IG_m$ is the $m$th Gini index and $\widehat{IG}_m$ is its  corresponding estimator (initially introduced by \cite{Vila2025}) given in  Remark \ref{rem-gini-index}.
\end{proposition}

\begin{remark}
	Note that Proposition \ref{prop-main} recovers a known result from \cite{Vila2025}.
\end{remark}

\section{Illustrative simulation study}\label{sec:04}

In this section, we present a Monte Carlo simulation study to evaluate the finite-sample performance of the proposed extended Gini index estimators defined in Equations~\eqref{estimator} and~\eqref{estimator-1}, which correspond to the sample-based versions of the extended lower and upper Gini indices introduced in Section~\ref{sec:02}.

The idea is to empirically assess for a particular case the theoretical unbiasedness of these estimators for gamma-distributed data, as established in Section~\ref{sec:03}. To this end, we simulate independent random samples of size $n \in \{10, 30, 50, 100, 200\}$ from a gamma distribution with shape parameter $\alpha = 2$ and rate parameter $\lambda = 1$. For each replication, we compute the extended lower Gini index estimator $\widehat{\,_iIG}_{m;\text{min}}$ and extended upper Gini index estimator $\widehat{\,^iIG}_{m;\text{max}}$ using fixed parameters $m = 3$ and $i = 3$, based on Equations~\eqref{estimator} and~\eqref{estimator-1}, respectively.

The simulation is repeated $N_{\text{sim}} = 500$ times for each sample size to estimate the empirical bias and mean squared error (MSE), computed as:
\begin{align*}
\widehat{\text{Bias}} = \frac{1}{N_{\text{sim}}} \sum_{\ell=1}^{N_{\text{sim}}} \left(\widehat{i\text{IG}}_m^{(\ell)} - \text{IG}_m \right),\quad
\widehat{\text{MSE}}  = \frac{1}{N_{\text{sim}}} \sum_{\ell=1}^{N_{\text{sim}}} \left(\widehat{i\text{IG}}_m^{(\ell)} - \text{IG}_m \right)^2,
\end{align*}
where $\widehat{i\text{IG}}_m^{(\ell)}\in\{\widehat{\,_iIG}_{m;\text{min}},\widehat{\,^iIG}_{m;\text{max}}\}$ is the $\ell$-th Monte Carlo replicate of the estimator computed using Equations~\eqref{estimator} or~\eqref{estimator-1}, and $\text{IG}_m\in\{{IG}_{m;\text{min}},{IG}_{m;\text{max}}\}$ is the corresponding theoretical value obtained from the expressions in Propositions~\ref{ext-gini-index} and~\ref{ext-gini-index-2}.

Table~\ref{tab:sim_results} summarizes the results of the Monte Carlo simulation study. For the extended lower Gini estimator $\widehat{\,_iIG}_{m;\text{min}}$, the empirical bias is small across all sample sizes considered, ranging from approximately $0.019$ (at $n = 10$) to $0.016$ (at $n = 200$). In line with the theoretical unbiasedness established in Section~\ref{sec:03}, the bias exhibits a clear decreasing trend as $n$ increases. Furthermore, the MSE decreases markedly with the sample size, from $0.00353$ at $n = 10$ to $0.00036$ at $n = 200$, which is consistent with the consistency of the estimator. For the extended upper Gini estimator $\widehat{\,^iIG}_{m;\text{max}}$, the empirical bias is negligible throughout, remaining very close to zero for all values of $n$ and alternating in sign. The MSE also decreases steadily with increasing sample size, from $0.00381$ at $n = 10$ to $0.00016$ at $n = 200$. Overall, the results show the good performance of both proposed estimators and corroborate the theoretical findings derived in Section~\ref{sec:03}.

\begin{table}[!ht]
\centering
\caption{Empirical bias and MSE of the extended lower and upper Gini estimators ($m=3$) based on 500 Monte Carlo replications under the Gamma(2,1) distribution.}
\centering
\begin{tabular}[t]{ccccccccccccccccc}
\toprule
$n$ & Bias & MSE & Estimator\\
\midrule
\cellcolor{gray!10}{10} & \cellcolor{gray!10}{0.01923} & \cellcolor{gray!10}{0.00353} & \cellcolor{gray!10}{Lower Gini}\\
30 & 0.02000 & 0.00117 & Lower Gini\\
\cellcolor{gray!10}{50} & \cellcolor{gray!10}{0.01725} & \cellcolor{gray!10}{0.00072} & \cellcolor{gray!10}{Lower Gini}\\
100 & 0.01569 & 0.00044 & Lower Gini\\
\cellcolor{gray!10}{200} & \cellcolor{gray!10}{0.01618} & \cellcolor{gray!10}{0.00036} & \cellcolor{gray!10}{Lower Gini}\\
\addlinespace
10 & -0.00369 & 0.00381 & Upper Gini\\
\cellcolor{gray!10}{30} & \cellcolor{gray!10}{-0.00137} & \cellcolor{gray!10}{0.00114} & \cellcolor{gray!10}{Upper Gini}\\
50 & -0.00006 & 0.00060 & Upper Gini\\
\cellcolor{gray!10}{100} & \cellcolor{gray!10}{-0.00116} & \cellcolor{gray!10}{0.00035} & \cellcolor{gray!10}{Upper Gini}\\
200 & 0.00019 & 0.00016 & Upper Gini\\
\bottomrule
\end{tabular}\label{tab:sim_results}
\end{table}

\section{Application to real data}
\label{sec:05}

In order to illustrate the practical utility of the proposed extended Gini index estimators defined in Equations~\eqref{estimator} and~\eqref{estimator-1}, we consider real-world income data. Specifically, we analyze data on gross domestic product (GDP) per capita (expressed in international-\$ at 2021 prices) for eleven South American countries in the year 2023; see \url{https://ourworldindata.org/grapher/gdp-per-capita-worldbank} and Table~\ref{table_countries}. The data reveal a marked degree of income dispersion across countries, ranging from approximately $9{,}844$ international-\$ (Bolivia) to $49{,}315$ international-\$ (Guyana), which motivates the use of inequality-sensitive measures such as those proposed in this paper.

\begin{table}[!ht]
\centering
\caption{2023 GDP per capita for South America countries.}
\begin{tabular}{lr}
  \hline
Countries ($n=11$) & GDP (international-\$ in 2021 prices) \\
  \hline
Guyana & 49315.16 \\
  Uruguay & 31019.31 \\
  Chile & 29462.64 \\
  Argentina & 27104.98 \\
  Suriname & 19043.71 \\
  Brazil & 19018.24 \\
  Colombia & 18692.38 \\
  Paraguay & 15783.11 \\
  Peru & 15294.26 \\
  Ecuador & 14472.32 \\
  Bolivia & 9843.97 \\
   \hline
\end{tabular}
\label{table_countries}
\end{table}

We fit a gamma distribution to the GDP per capita data utilizing the \texttt{fitdistrplus} package in \texttt{R} \citep{fitdistrplus:15}. The diagnostic plots shown in Figure~\ref{fig:diagnostic_plot} indicate that the gamma distribution provides a reasonable fit to the data. To further assess this, we apply the Kolmogorov-Smirnov (KS) and Cramér-von Mises (CvM) goodness-of-fit tests, obtaining $p$-values of $0.508$ and $0.784$, respectively, which provide no evidence against the gamma distribution assumption. Overall, these findings support the use of the gamma model for the income data and justify the application of the estimators established in Section~\ref{sec:03}.

\begin{figure}[H]
\centering
\includegraphics[width=0.8\textwidth]{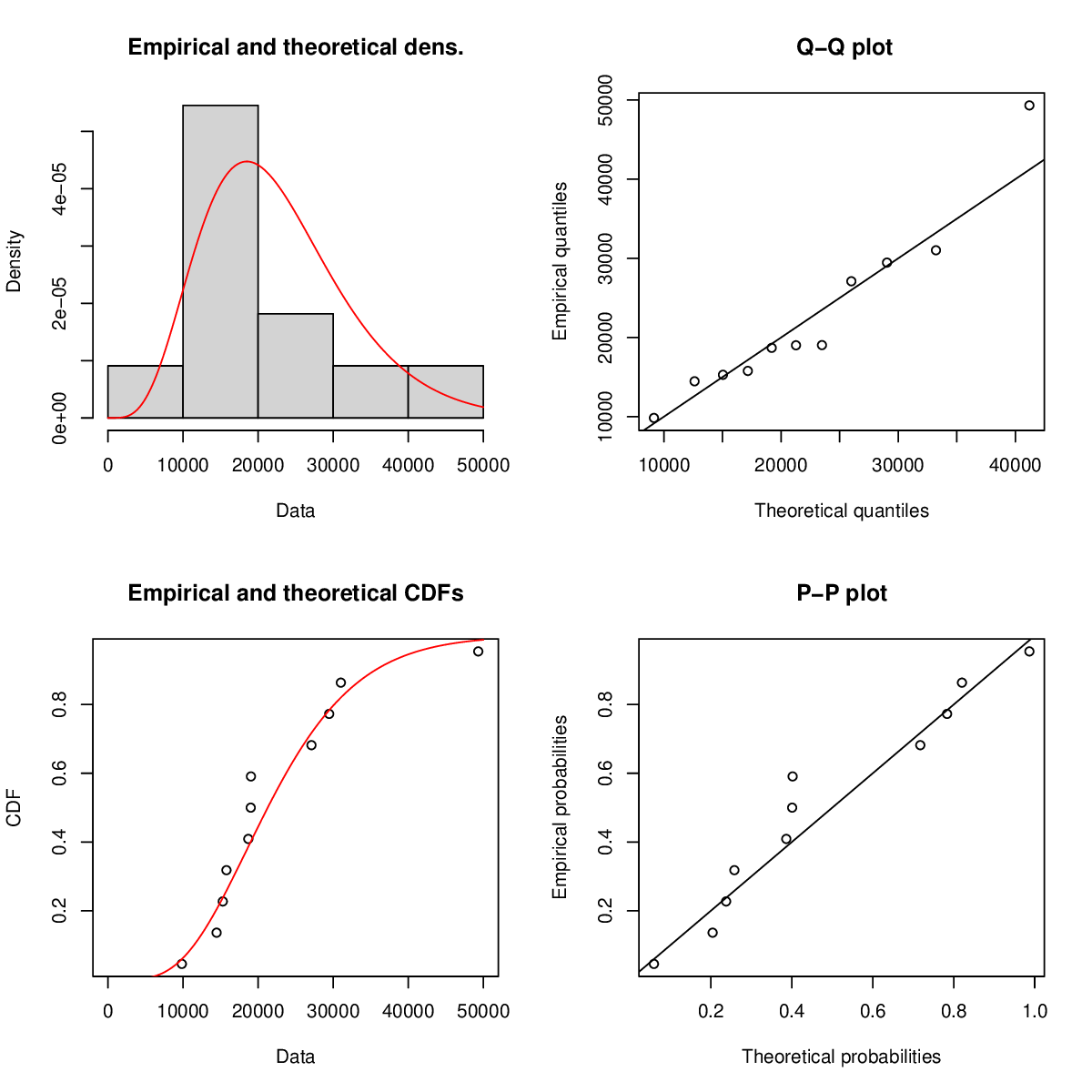}
\caption{Diagnostic plots based on the gamma distribution applied to 2023 GDP per capita data.}
\label{fig:diagnostic_plot}
\end{figure}

Figures~\ref{fig:heatmap_lower_gini} and~\ref{fig:heatmap_upper_gini} display heatmaps of the estimated extended lower and upper Gini indices, respectively, computed from the 2023 GDP per capita data for South American countries over a grid of parameter values $(m, i)$.

Figure~\ref{fig:heatmap_lower_gini} reveals a clear and systematic pattern in the estimated extended lower Gini index $\widehat{\,_iIG}_{m;\text{min}}$: for fixed $m$, the estimates increase monotonically as $i$ grows, which reflects the fact that larger values of $i$ compare observations that are progressively more distant from the minimum order statistic, thereby capturing a greater degree of spread in the lower tail of the income distribution. Furthermore, for fixed $i$, the estimates tend to increase with $m$, since larger values of $m$ amplify the expected gap between the reference minimum and the $i$th order statistic in a sample of that size. It is worth mentioning that this joint sensitivity to both $m$ and $i$ allows us to modulate the emphasis placed on lower-tail inequality in a flexible, parameter-driven manner, a feature that is not available in classical Gini-type measures.
\begin{figure}[H]
\centering
\includegraphics[width=0.7\textwidth]{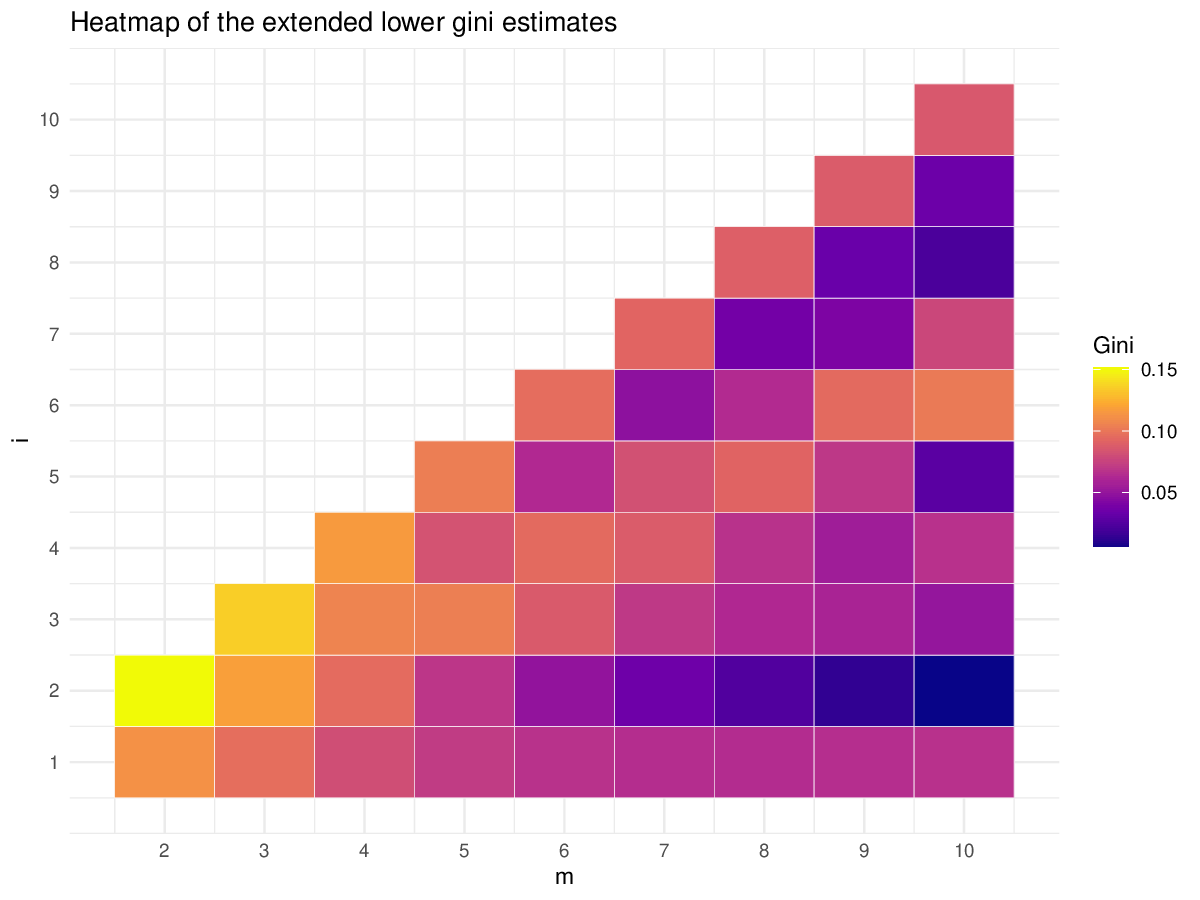}
\caption{Heatmap of the estimates of the extended lower Gini index estimator $\widehat{\,_iIG}_{m;\text{min}}$ for different values of $m$ and $i$, based on 2023 GDP per capita data.}
\label{fig:heatmap_lower_gini}
\end{figure}
\begin{figure}[H]
\centering
\includegraphics[width=0.7\textwidth]{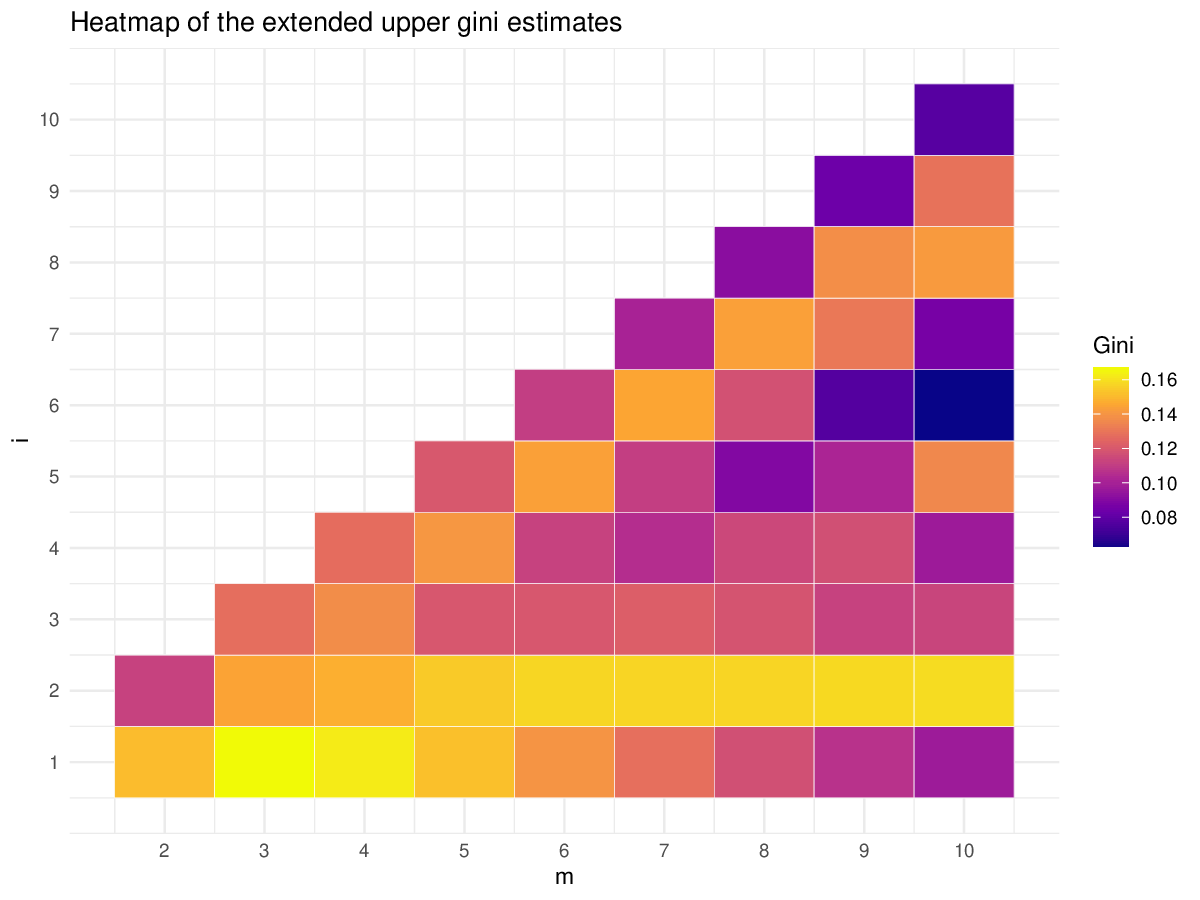}
\caption{Heatmap of the estimates of the extended upper Gini index estimator $\widehat{\,^iIG}_{m;\text{max}}$ for different values of $m$ and $i$, based on 2023 GDP per capita data.}
\label{fig:heatmap_upper_gini}
\end{figure}

Figure~\ref{fig:heatmap_upper_gini} presents the analogous results for the extended upper Gini index $\widehat{\,^iIG}_{m;\text{max}}$. In general, the estimated values are larger than those of the corresponding lower index for the same parameter configurations, which is consistent with the fact that the maximum order statistic constitutes a more extreme reference point than the minimum, resulting in larger expected pairwise differences and, consequently, a higher measured degree of inequality. The gradient structure of the heatmap mirrors that of Figure~\ref{fig:heatmap_lower_gini} in qualitative terms, with estimates increasing as $m$ and $i$ grow, whereas the overall level of the index is shifted upward relative to the lower case, reflecting the greater sensitivity of the upper index to the concentration of income at the top of the distribution. Overall, the results show that the extended lower and upper Gini indices yield a richer characterization of income inequality among South American countries than a single aggregate coefficient, making it possible to assess distributional features at both the lower and upper tails simultaneously.

\section{Concluding remarks}\label{sec:06}
%

{
We introduced two flexible extensions of the classical Gini index, referred to as the extended lower and extended upper Gini indices. The proposed measures are based on the differences between an observation and the minimum and maximum order statistics in samples of size $m$, thereby providing a position-oriented perspective on inequality. Unlike conventional Gini-type measures, they allow inequality to be assessed relative to the lower and upper tails of the distribution and offer a more detailed characterization of the distribution of income disparities.

From an inferential perspective, we established the consistency and asymptotic normality of the proposed estimators under mild regularity conditions. Furthermore, for gamma-distributed populations, we derived closed-form expressions for their expectations and proved their unbiasedness, thereby extending previous results available in the literature. Monte Carlo simulation studies corroborated the theoretical findings and demonstrated the satisfactory finite-sample performance of the estimators.

The methodology was illustrated through an analysis of 2023 GDP per capita data from South American countries. The empirical results showed that the extended lower and upper Gini indices provide a richer spectrum of inequality measures than traditional Gini-type indices, revealing features of the distribution that may remain hidden when inequality is summarized by a single aggregate coefficient. These findings suggest that the proposed indices constitute useful complementary tools for inequality analysis and may be fruitfully applied in a variety of economic and social contexts.
}

\paragraph*{Acknowledgements}
The research was supported in part by CNPq and CAPES grants from the Brazilian government.

\paragraph*{Disclosure statement}
There are no conflicts of interest to disclose.

\paragraph*{Data availability statement}
The data used in this study are publicly available and can be accessed online at the following link: \url{https://ourworldindata.org/grapher/gdp-per-capita-worldbank}

\paragraph*{Author contribution statement}
Roberto Vila: Conceptualization, Methodology, Formal analysis, Writing-original draft.
Helton Saulo: Supervision, Validation, Writing-review \& editing.



\end{document}